\documentclass{article}
\usepackage{arxiv}
\usepackage[english]{babel}
\usepackage[utf8]{inputenc} 
\usepackage[T1]{fontenc}    
\usepackage{hyperref}       
\usepackage{url}            
\usepackage{booktabs}       
\usepackage{amsfonts}       
\usepackage{nicefrac}       
\usepackage{microtype}      
\usepackage{graphicx}
\usepackage[sort, numbers]{natbib}
\usepackage{doi}
\usepackage{amsmath}
\usepackage{amsthm}
\usepackage[english]{babel}
\usepackage[dvipsnames]{xcolor}
\newtheorem{theorem}{Theorem}[section]
\newtheorem{remark}[theorem]{Remark}
\newcommand{\hl}[1] {\textcolor{black}{#1}}

\newcommand{\p}{\partial}
\newcommand{\T}{\theta}

\newcommand{\n}{\hat{n}}

\title{An Entropic Approach to Classical Density Functional Theory}
\author{ \href{https://orcid.org/0000-0000-0000-0000}{\includegraphics[scale=0.06]{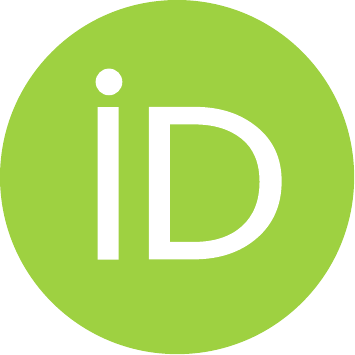}\hspace{1mm}Ahmad Yousefi} \\
	Department of Physics\\
	University at Albany-SUNY\\
	Albany, NY 12222 \\
	\texttt{ayousefi@albany.edu} \\
	\texttt{ahmad369@gmail.com} \\
	\And
	\href{https://orcid.org/0000-0000-0000-0000}{\includegraphics[scale=0.06]{orcid.pdf}\hspace{1mm}Ariel Caticha} \\
	Department of Physics\\
	University at Albany-SUNY\\
	Albany, NY 12222 \\
	\texttt{acaticha@albany.edu} \\
}

\hypersetup{
pdftitle={A template for the arxiv style},
pdfsubject={q-bio.NC, q-bio.QM},
pdfauthor={David S.~Hippocampus, Elias D.~Striatum},
pdfkeywords={First keyword, Second keyword, More},
}
\numberwithin{equation}{section}
\begin{document}
\maketitle
\begin{abstract}
	The classical Density Functional Theory (DFT) is introduced as an application of entropic inference for inhomogeneous fluids in thermal equilibrium.  It is shown that entropic inference reproduces the variational principle of DFT when information about expected density of particles is imposed. This process introduces a family of trial density-parametrized probability distributions, and consequently a trial entropy, from which the preferred one is found using the method of Maximum Entropy (MaxEnt).  As an application, the DFT model for slowly varying density is provided, and its approximation scheme is discussed.
\end{abstract}
\section{Introduction}
\hl{The} Density Functional Theory was first developed in the context of quantum mechanics and only later extended to the classical regime. The theory was first introduced by Kohn and Hohenberg (1964)\cite{kohn} as a computational tool to calculate the spatial density of an electron gas in the presence of an external potential at zero temperature. Soon afterwards, Mermin provided the extension to finite temperatures \cite{mermin}. Ebner, Saam, and Stroud (1976)\cite{ebner} applied the idea to simple classical fluids, and Evans (1979) provided a systematic formulation in his classic paper \cite{evans79} "The nature of the liquid-vapour interface and other topics in the statistical mechanics of non-uniform, classical fluids.".\\
\hl{
The majority of physicists and chemists today are aware of the quantum DFT and the Kohn-Sham model \cite{kohnsham}, while fewer are familiar with the classical DFT, a historical review of quantum DFT and its vast variety of applications is found in \cite{halfcentury,rise}. The classical DFT, similarly, is a "formalism designed to tackle the statistical mechanics of inhomogeneous fluids"\cite{newdev}, which has been used to investigate a wide variety of equilibrium phenomena, including surface tension, adsorption, wetting, fluids in porous materials, and the chemical physics of solvation.   
}
\\
\hl{Just like the Thomas-Fermi-Dirac theory is usually regarded as a precursor of quantum DFT, the van der Waals' thermodynamic theory of capillarity under the hypothesis of a continuous variation of density \cite{van} can be regarded as the earliest work on classical DFT, without a fundamental proof of existence for such a variational principle.   
}
\\
\hl{
"The long-term legacy of DFT depends largely on the
continued value of the DFT computer programs that
practitioners use daily."\cite{halfcentury} The algorithms behind the computer programs, all starting from an original Hartree-Fock method to solve the N-particle Schr\"odinger equation, have evolved by many approximations and extensions implemented over time by a series of individuals, although the algorithms produce accurate results, they do not mention the HK variational principle. Without the variational principle the computer codes are suspect of being ad hoc or intuitively motivated without a solid theoretical foundation, therefore, the DFT variational principle not only scientifically justifies the DFT algorithms, but also provides us with a basis to understand the repeatedly modified algorithms behind the codes. 
}\\
In this work we derive the classical DFT as an application of the method of maximum entropy \cite{Jaynes1,Jaynes2,wdws,psm,entropicphysics}. \hl{This integrates the classical DFT with other formalisms of classical statistical mechanics (canonical, grand canonical, etc.) as an application of information theory. Our approach not only enables one to understand the theory from the Bayesian point of view, but also provides a framework to construct equilibrium theories on the foundation of MaxEnt. We emphasize that our goal is not derive an alternative to DFT. Our goal is purely conceptual. We wish to find a new justification or derivation of DFT that makes it explicit how DFT fits within the MaxEnt approach to statistical mechanics. The advantage of such an understanding is the potential for future applications that are outside the reach of the current version of DFT.}
\\
In section \ref{cinference}, we briefly review entropic inference as an inference tool which updates probabilities as degrees of rational belief in response to new information. Then we show that any entropy maximization, produces a contact structure which is invariant under \hl{the} Legendre Transformations; this enables us to take advantage of these transformations for maximized entropy functions (functionals here) found from constraints other than those of thermal equilibrium as well as thermodynamic potentials.\\
In section \ref{preferred} we briefly review the method of relative entropy for optimal approximation of probabilities, \hl{which allows us to derive and then generalize the Bogolyubov variational principle.} Then we apply it for the especial case wherein the trial family of probabilities are parametrized by the density function $n(x)$.  
\\
In section \ref{cdensityformalism}, the Density Functional formalism is introduced as an extension of the existing ensemble formalisms of statistical mechanics (canonical, grand canonical, etc.) and we show that the core \hl{DFT theorem} is an immediate consequence of MaxEnt: we prove that in the presence of \hl{an} external potential $v(x)$, there exists \hl{a trial} density functional entropy $S_v(E; n]$ maximized at the equilibrium density. We also prove that this entropy maximization is equivalent to minimization of a density functional potential $\Omega(\beta;n]$ given by
\begin{equation}
    \Omega_v(\beta;n]=\int d^3x v(x)n(x)+F(\beta;n]
\end{equation}
where $F(\beta;n]$ is independent of $v(x)$. 
This formulation achieves two objectives. \textbf{i)} It shows that \hl{the} density functional variational principle is an application of \hl{MaxEnt} for non-uniform fluids at equilibrium and therefore, varying \hl{the} density $n(x)$ in $S_v(E;n]$ \textbf{does not} imply that the functional represents entropy of any non-equilibrium system, this \hl{trial} entropy, although very useful, is just a mathematical construct which allows us to incorporate constraints which are related to one another by definition. \textbf{ii)} By this approach we show that \hl{the} Bayesian interpretation of probability liberates \hl{the fundamental theorem of the} DFT from \hl{an} imaginary grand-canonical ensemble, i.e. \hl{the} thermodynamic chemical potential is appropriately defined without need to define microstates for varying number of particles.\\
\hl{Finally, in section \ref{gradient}, as an illustration we discuss the already well-known example of a slowly varying inhomogeneous fluid.} We show that our entropic DFT allows us \hl{to} reproduce the gradient approximation results derived by Evans \cite{evans79}. There are two different approximations involved, \textbf{i)} \hl{rewriting} \hl{the} non-uniform direct correlation function \hl{in terms of} \hl{the} uniform one; and \textbf{ii)} the use of linear response theory to evaluate Fourier transform of direct correlation functions. The former assumes that \hl{the} density is uniform inside each volume element, and the latter assumes that difference of densities for neighboring volume elements is small compared to their average. 
\section{Entropic inference}\label{cinference}
A discussion of the method of maximum entropy as a tool for inference is found in \cite{entropicphysics}. Given a \textbf{prior} Probability Distribution Function (\textbf{PDF}) $Q(X)$, we want to find the \textbf{posterior} PDF $P(X)$ subject to constraints on expected values of functions of $X$.\\
Formally, we need to maximize \hl{the} relative entropy
\begin{equation}
    S_r[P|Q]\equiv-\sum_X (PlogP-PlogQ) \; ,
\end{equation}
under constraints $A_i=\sum_X P(X) \hat A_i(X)$ and $1=\sum_X P(\hl{X})$;
where $A_i$'s are real numbers, and $\hat A_i$'s are real-valued functions on \hl{the} space of $X$, and $1 \leq i \leq m$ for $m$ number of constraints.\\
The maximization process yields \hl{the} posterior probability 
\begin{equation}
    P(X)=Q(X)\frac{1}{Z}e^{-\sum_i\alpha_i \hat A_i(X)} \; , \quad \quad \text{where }  
    Z=\sum_X Q(X)e^{-\sum_i\alpha_i \hat A_i(X)}\;,
\end{equation}
and $\alpha_i$'s are Lagrange multipliers associated with $A_i$'s.\\
Consequently the maximized entropy is
\begin{equation}\label{ent101}
    S=\sum_i\alpha_i A_i+log Z\; .
\end{equation}
Now we can show that the above entropy readily produces a contact structure, we can calculate \hl{the} complete differential of equation (\ref{ent101}) to define \hl{the} vanishing one-form $\omega_{cl}$ as
\begin{equation}
    \omega_{cl}\equiv dS-\sum_i \alpha_i dA_i=0 \; .
\end{equation}
Therefore any classical entropy maximization with $m$ constraints produces a contact structure $\lbrace  \hl{\mathbb T},\omega_{cl}\rbrace$ in which manifold $\mathbb T$ has $2m+1$ coordinates $\lbrace q_0,q_1,\dots q_m,p_1,\dots,p_m\rbrace$.\\
The physically relevant manifold $\mathbb M$ is an $m$-dimensional sub-manifold of $\mathbb T$, on which $\omega_{cl}$ vanishes; i.e. $\mathbb M$ is determined by $1+m$ equations
\begin{equation}
    q_0\equiv S(\lbrace A_i\rbrace)\; , \quad \quad p_i\equiv\alpha_i=\frac{\p S}{\p A_i}\; .
\end{equation}
Legendre Transformations defined as, 
\begin{equation}
    q_0\longrightarrow  \; q_0 -\sum_{j=1}^{l}p_jq_j\; ,
\end{equation}
\begin{equation}
    \nonumber
    q_i\longrightarrow  \; p_i \; ,\quad p_i\longrightarrow \; -q_i \; ,  \quad\quad \text{for } 1\leq i\leq l \; ,
\end{equation}
are coordinate transformations on space $\mathbb T$ under which $\omega_{cl}$ is conserved. It has been shown \cite{rajeev,HSTh} that the laws of thermodynamics produce a contact structure conforming to \hl{the} above prescription, here we are emphasizing that the contact structure is an immediate consequence of \hl{MaxEnt}, and therefore it can be utilized in applications of information theory, beyond thermodynamics.
\section{Maxent and Optimal Approximation of Probabilities}\label{preferred}
The posterior PDF found from entropic inference, is usually too complicated to be used for practical purposes. A common solution is to approximate the posterior PDF by \hl{a} more tractable family of PDFs $\lbrace p_\T\rbrace$ \cite{catichatseng}. Given the exact probability $p_0$, the preferred member of tractable family $p_{\T^*}$ is found by \hl{maximizing the entropy of $p_{\T}$ relative to $p_0$:}
\begin{equation}\label{optimal}
    \frac{\delta S_r[p_\T|p_0]}{\delta \T}\Big|_{\T=\T^*}=0 \;.
\end{equation}
\hl{The} density functional formalism is \hl{a} systematic method in which the family of trial probabilities is parametrized by \hl{the} density of particles; in section \ref{cdensityformalism} we shall use the method of maximum entropy to determine the family of trial distributions parametrized by $n(x)$, $p_\T\equiv p_{n}$. So that we can rewrite equation \ref{optimal} as
\hl{
\begin{equation}\label{optimaln}
    \frac{\delta}{\delta n(x')}
    \Bigg[
       S_r[p_{n}|p_0]+\alpha_{eq}[N-\int d^3x n(x)]
    \Bigg]_{n(x)=n_{eq}(x)}=0 \;.
\end{equation}
}
We will see that the canonical distribution itself is a member of the trial family, therefore in this case, the exact solution to equation (\ref{optimaln}) is $p_0$ itself\hl{:}  
\begin{equation}
    p_{n}\Big|_{n(x)=n_{eq}(x)}=p_0\; .
\end{equation}
\section{Density Functional Formalism}\label{cdensityformalism}
An equilibrium formalism of statistical mechanics is a relative entropy maximization process consisting of three crucial \hl{elements}: i) \hl{One must choose the microstates that describe the system of inference.} ii) \hl{The prior is chosen to be uniform.} iii) \hl{One must select the constraints that represent the information that is relevant to the problem in hand}.\\
In the density ensemble, microstates of the system are given as positions and momenta of all $N$ particles of the same kind, given \hl{the} uniform prior probability distribution
\begin{equation}
   Q(\lbrace \vec x_1,\dots,\vec x_N;\vec p_1,\dots,\vec p_N\rbrace)=constant.
\end{equation}
Having in mind that we are looking for \textbf{thermal properties of inhomogeneous fluids}, it is natural to choose the density of particles $n(x)$ as computational constraint, and the expected energy $E$ as thermodynamic constraint, in which $n(x)$ represents the \textbf{inhomogeneity} and $E$ defines \hl{the} thermal equilibrium.\\
Note that all constraints (computational, thermal, etc.) in the framework can be incorporated as inferential constraints and can be imposed as prescribed in section \ref{cinference}.\\   
The density \hl{constraint} holds for every point in space, therefore we have $1+1+\mathbb R^3$ constraints, one for normalization, one for total energy, and one for  density of particles at each point in space; so we have to maximize \hl{the} relative entropy
\begin{equation}\label{sr}
   S_r[P|Q]=-\frac{1}{N!}\int d^{3N}xd^{3N}p (PlogP-PlogQ)
   \equiv -Tr_c (PlogP-PlogQ),
\end{equation}
subject to constraints
\begin{subequations}\label{gheyds}
\begin{align}
    1=&\langle 1 \rangle, \quad
    E=\langle \hat H_v\rangle, \\    
    n(x)=&\langle \hat n_x\rangle \quad \text{where} \int d^3x n(x)=N,
\end{align}
\end{subequations}
where $\langle . \rangle\equiv \frac{1}{N!}\int (.)Pd^{3N}xd^{3N}p$. \hl{The} classical Hamiltonian operator $\hat H$ and the particle density operator $\hat n_x$ are given as
\begin{equation}
    \hat H_v\equiv\sum_{i=1}^N
    v(x_i)
    +\hat K(p_1,\dots,p_N)+\hat U(x_1,\dots,x_N),
\end{equation}
\begin{equation}\label{densityop}
 \hat n_x\hl{\equiv}\sum_i^N \delta(x-x_i).   
\end{equation}
The density $n(x)$ is not an arbitrary function; it is constrained by a fixed total number of particles,
\begin{equation}\label{secondary}
    \int d^3x n(x)=N.
\end{equation}
Maximizing (\ref{sr}) subject to (\ref{gheyds}) gives the posterior probability $P(x_1,\dots,x_N;p_1,\dots,p_N)$ as
\begin{equation}\label{postprob}
    P=\frac{1}{Z_v}e^{-\beta \hat H_v-\int d^3x \alpha(x)\hat n_x}\quad\text{under condition}  \int d^3xn(x)=N.
\end{equation}
where $\alpha(x)$ and $\beta$ are \hl{Lagrange} multipliers.\\
\hl{The Lagrange multiplier function $\alpha(x)$ is implicitely determined by
\begin{equation}
    \frac{\delta logZ_v}{\delta \alpha(x)}=-n(x)
\end{equation}
and by equation (\ref{secondary})
\begin{equation}
    -\int d^3x \frac{\delta logZ_v}{\delta \alpha(x)}=N \; .
\end{equation}
}
Substituting the trial probabilities from (\ref{postprob}) into (\ref{sr}) gives the \hl{\textbf{trial}} entropy $S_v(E\hl{;}n]$ as
\begin{equation}\label{entopen}
    S_v(E;n]=\beta E+\int d^3x \alpha(x) n(x)+log Z_v,
\end{equation}
where $Z_v(\beta;\alpha]$ is the \hl{\textbf{trial}} partition function defined as
\begin{equation}
    Z_v(\beta;\alpha]=Tr_c e^{-\beta \hat H_v-\int d^3x \alpha(x)\hat n_x} \; .
\end{equation}
The equilibrium density \hl{$n_{eq}(x)$} is that which maximizes $S_v(E\hl{;}n]$ subject to $\int d^3xn(x)=N$:
\begin{equation}\label{f2cb}
    \frac{\delta}{\delta n(x')}\Bigg[
       S_v(E;n]+\alpha_{eq}[N-\int d^3xn(x)]
    \Bigg]=0\quad \textit{for fixed E}.
\end{equation}
Next, perform a Legendre transformation and define \hl{the} Massieu functional $\tilde S_v(\beta,n]$ as
\begin{equation}\label{massieu}
    \tilde S_v\equiv S_{\hl{v}}-\beta E,
\end{equation}
so that we can rewrite  equation (\ref{f2cb}) as
\begin{equation}\label{massvar}
    \frac{\delta}{\delta n(x')}\Bigg[
       \tilde S_v(\beta;n]-\alpha_{eq}\int d^3xn(x)
    \Bigg]=0 \quad \textit{for fixed $\beta$} \; .
\end{equation}
\hl{Combine (\ref{entopen}), (\ref{massieu}), and (\ref{massvar}), and use the variational derivative identity $\frac{\delta n(x)}{\delta n(x')}=\delta(x-x')$ to find}
\begin{equation}\label{gold00}
    \int d^3x \Bigg[
       \frac{\delta logZ_v(\beta;\alpha]}{\delta \alpha(x)}+n(x)
    \Bigg]\frac{\delta\alpha(x)}{\delta n(x')}
    =\alpha_{eq}-\alpha(x').
\end{equation}
\hl{The LHS of equation (\ref{gold00}) vanishes by (\ref{secondary}) and therefore the RHS must vanish for an arbitrary $n(x)$ which implies that}
\begin{equation}\label{alpha}
       \alpha(x)=\alpha_{eq}\; , \quad \text{and} \quad \frac{\delta logZ_v}{\delta \alpha(x)}\hl{\Big|_{\alpha_{eq}}}=-n_{\hl{eq}}(x) \; .
\end{equation}
Substituting (\ref{alpha}) into (\ref{postprob}) yields the equilibrium probability distribution
\begin{equation}\label{finalpost}
   P^*(x_1,\dots,x_N;p_1,\dots,p_N)=\frac{1}{Z_v^*}e^{-\beta\hat H_v-\alpha_{eq}\int d^3x\hat n_x}
   =
   \frac{1}{Z_v^*}e^{-\beta\hat H_v-\alpha_{eq}N}
\end{equation}
where $Z_v^*(\beta,\alpha_{eq})=Tr_c e^{-\beta \hat H_v-\alpha_{eq}N}.\label{zstar}$\\
From the inferential point of view, \textbf{the variational principle for the grand potential and the equilibrium density}\cite{evans79} is proved at this point; we showed that for an arbitrary classical Hamiltonian $\hat H_v$, there exists \hl{a trial} entropy $S_v(E;n(x)]$ defined by equation (\ref{entopen}), which assumes its maximum at fixed energy and varying $n(x)$ under \hl{the} condition $\int d^3x n(x)=N$ at \hl{the} equilibrium density, and gives \hl{the} posterior PDF (equation \ref{finalpost}) equal to that of \hl{the} canonical distribution.\\   
The massieu function $\tilde S_v(\beta;n(x)]$ from equation (\ref{massieu}) defines \hl{the} \textbf{density functional potential} $\Omega_v(\beta;n(x)]$ by
\begin{equation}\label{Omega}
    \Omega_v(\beta;n]\equiv\frac{-\tilde S_v(\beta;n]}{\beta}
    =-\int d^3x\frac{\alpha(x)}{\beta}n(x)-\frac{1}{\beta}logZ_v(\beta;\alpha] \; , 
\end{equation}
so that the maximization of $S_v(E;n]$ (\ref{entopen}) in the vicinity of equilibrium, is equivalent to the minimization of $\Omega_v(\beta;n(x)]$ (\ref{Omega}) around the same equilibrium point
\begin{equation}\label{asli}
    \frac{\delta}{\delta n(x')}
    \Bigg[ \Omega_v(\beta;n]+\frac{\alpha_{eq}}{\beta}\int d^3xn(x)
    \Bigg]=0 \; .
\end{equation}
After we find $\Omega_v$, we just need to recall that $\alpha(x)=-\beta\frac{\delta \Omega_v}{\delta n(x)}$ and substitute in equation (\ref{alpha}) to \hl{recover the} "core integro-differential equation"\cite{evans79} of DFT as 
\begin{equation}\label{core}
    \nabla\Big(
    \frac{\delta \Omega(\beta;n]}{\delta n(x)}
    \Big)_{eq}=0 
\end{equation}
which implies that
\begin{equation}\label{omvar}
    \Omega_{v;eq} \leq \Omega_v(\beta;n],
\end{equation}
 where
\begin{equation} \label{ome}
    \Omega_{v,eq}(\beta;n]=-\frac{\alpha_{eq}}{\beta}\int d^3x n(x)-\frac{1}{\beta}log Z_v^*(\beta,\alpha_{eq}).
\end{equation}
From equation (\ref{Omega}) it is clear that
\begin{equation}\label{om}
    \Omega_v(\beta;n]=\int d^3x v(x)n(x)+\langle \hat K+\hat U\rangle-\frac{S_v(E;n]}{\beta} \; .
\end{equation}
It is convenient to define \hl{the} \textbf{intrinsic density functional potential} $F_v$ as
\begin{equation}\label{intrins}
    F_v(\beta;n]\equiv \langle \hat K+\hat U\rangle-\frac{S_v}{\beta}
\end{equation}
to have
\begin{equation}\label{om1}
    \Omega_v(\beta;n]=\int v(x)n(x)+F_v(\beta;n]\; .
\end{equation}
\hl{Now we are ready to restate the fundamental theorem of the classical DFT:}
\begin{theorem}
The intrinsic functional potential $F_v$ is a functional of density $n(x)$ and is independent of the external potential:
   \begin{equation}
       \frac{\delta F_v(\beta;n]}{\delta v(x')}=0 \quad \text{for fixed $\beta$ and $n(x)$.}
   \end{equation}
\end{theorem}
\begin{proof}
    The crucial observation behind the DFT formalism is that $P$ and $Z_v$ depend on the external potential $v(x)$ and the Lagrange multipliers $\alpha(x)$ only through the particular combination $\bar \alpha(x)\equiv \beta v(x)+\alpha(x)$. \hl{Substitute} equation (\ref{entopen}) in (\ref{intrins}) to get
    \begin{equation}\label{intrins1}
        \beta F_v(\beta;n]=logZ(\beta;\bar \alpha]+\int d^3x \bar \alpha(x)n(x),
    \end{equation}
    where $Z(\beta;\bar \alpha]=Z_v(\beta;\alpha]=Tr_c e^{-\beta(\hat K+\hat U)-\int d^3x\bar \alpha(x)\n_x} \; .$
    The functional derivative of $\beta F_v$ at fixed $n(x)$ and $\beta$ is
    \begin{equation}
        \frac{\delta (\beta F_v(\beta;n])}{\delta v(x')}\Big|_{\beta,n(x)}
        =
        \int d^3x''\frac{\delta}{\delta\bar\alpha (x'')} 
        \Big[
           log Z(\beta;\bar\alpha]+\int d^3x \bar\alpha(x)n(x) 
        \Big]
        \frac{\delta \bar\alpha(x'')}{\delta v(x')}\Big|_{\beta,n(x)} \; .
    \end{equation}
    Since $n(x')=-\frac{\delta logZ(\beta;\bar\alpha]}{\delta \bar\alpha(x')}$, keeping $n(x)$ fixed is achieved by keeping $\bar\alpha(x)$ fixed:
    \begin{equation}
        \frac{\delta \bar\alpha(x'')}{\delta v(x')}\Big|_{\beta,n(x)}=\frac{\delta \bar\alpha(x'')}{\delta v(x')}\Big|_{\beta,\bar\alpha(x)}=0 \; ,
    \end{equation}
    so that
    \begin{equation}
        \frac{\delta F_v(\beta,n]}{\delta v(x')}
        \Big|_{\beta,n(x)}
        =0\; ,
    \end{equation}
    which concludes the proof; thus, we can write down the intrinsic potential as
    \begin{equation}
        F(\beta,n(x)]=F_v(\beta,n(x)]\; .
    \end{equation}
\end{proof}
\begin{remark}
    Note that \hl{since a} change of the external potential $v(x)$ can be compensated by a suitable change of the multiplier $\alpha(x)$ in such a way as to keep $\bar\alpha(x)$ fixed, such changes in $v(x)$ will have no effect on $n(x)$. Therefore keeping $n(x)$ fixed on the left hand side of (\ref{intrins1}) means that $\bar\alpha(x)$ on the right side is fixed too.
\end{remark}
Now we can substitute equation (\ref{om1}) in (\ref{asli}), and define \hl{the} chemical potential 
\begin{equation}
    \mu\equiv\frac{-\alpha_{eq}}{\beta} \; ,
\end{equation}
to have
\begin{equation}\label{aslitar}
    \frac{\delta}{\delta n(x')}
    \Bigg[ \int d^3x v(x)n(x)+F(\beta;n]-\mu\int d^3xn(x)
    \Bigg]_{n(x)=n_{eq}(x)}=0 \; .
\end{equation}
We can also substitute (\ref{om1}) in (\ref{core}) to find
\begin{equation}\label{potential}
    v(x)+\frac{\delta F}{\delta n(x)}\Big|_{eq}=\mu,
\end{equation}
which allows us to \hl{define and} \textbf{interpret} $\mu_{in}(x;n]\equiv \frac{\delta F}{\delta n(x)}$ as \hl{the} \textbf{intrinsic chemical potential} of the system. 
To proceed further we also split $F$ into that of ideal gas plus \hl{the} interaction part as
\begin{equation}
    F(\beta;n]=F_{id}(\beta;n]-\phi(\beta;n].
\end{equation}
Differentiating with $\frac{\delta}{\delta n(x)}$ gives
\begin{equation}
    \beta\mu_{in}(x;n]=log(\lambda^3n(x))-c(x;n] \; ,
\end{equation}
where for \hl{a} monatomic gas $\lambda=\Big(\frac{2\pi\hbar^2}{m}\Big)^{1/2}$. \hl{The} additional one body potential $c(x;n]=\frac{\delta \phi}{\delta n(x)}$ is related to \hl{the} Ornstein-Zernike direct correlation function of non-uniform fluid \cite{18,19,22} by
\begin{equation}
    c^{(2)}(x,x';n]\equiv\frac{\delta c(x;n]}{\delta n(x')}=\frac{\delta^2\phi(\beta;n]}{\delta n(x) \delta n(x')} \; .
\end{equation}
\section {Slowly Varying Density and Gradient Expansion}\label{gradient}
We have proved that the solution to equation (\ref{aslitar}) is the equilibrium density. But the functional $F(\beta;n]$ needs to be approximated because \hl{the} direct calculation of $F$ involves \hl{calculating} the canonical partition function, the task which we have been avoiding to begin with. Therefore different models of DFT may vary in their approach for guessing $F(\beta;n]$.
Now assume that we are interested in a monatomic fluid with a slowly varying external potential. In our language, it means that we use \hl{the} approximation $\int d^3x \equiv \sum (\Delta x)^3$ where $\Delta x$ is much longer than the density correlation length and the change in density in each volume element is small compared to average density.
\hl{This allows us to interpret each volume element $(\Delta x)^3$ as a fluid at grand canonical equilibrium with the rest of the fluid as its thermal and particle bath.}\\ 
Similar to \cite{evans79}, we expand $F(\beta;n]$ as
\begin{equation}\label{e311}
    F(\beta;n]=\int d^3x \Big[f_0(n(x))+f_2(n(x))|\nabla n(x)|^2+\mathcal O(\nabla^4 n(x))\Big] \; .
\end{equation}
Differentiating with respect to $n(x)$ we have
\begin{equation}\label{muin}
    \mu_{in}(x;n]=\frac{\delta F}{\delta n(x)}
    =f'_0(n(x))-f'_2(n(x))|\nabla n(x)|^2-2f_2(n(x))\nabla^2n(x).
\end{equation}
In \hl{the} absence of \hl{an} external potential, $v(x)=0$, the second and the third terms in \hl{the} RHS of (\ref{muin}) vanish, and also \hl{from equation (\ref{potential})}, \hl{$\mu_{in}=\mu$} , therefore we have
\begin{equation}\label{f0}
    f_0'(n)=\mu(n(x)),
\end{equation}
where $\mu(n(x))$ is chemical potential of \hl{a} uniform fluid with density $n=n(x)$.
\hl{On} the other hand\hl{,} with the assumption that each volume element behaves as if it is in grand canonical ensemble for itself under influence of both external potential and additional one body interaction $c(x;n]$ we know that \hl{the} second derivative of $F$ is related to Ornstein-Zernike theory by
\begin{equation}
    \beta\frac{\delta^2 F}{\delta n(x)\delta n(x')}=\frac{\delta (x-x')}{n(x)}-c^{(2)}(x,x';n] \; .
\end{equation}
Therefore we have a Taylor expansion of $F$ around \hl{the} uniform density as
\begin{align}\label{taylor}
    F[n(x)]=&F[n]+\int d^3x \Big[\frac{\delta F}{\delta n(x)}\Big]_{n_{eq}(x)}\tilde n(x)\\ \nonumber
    &+\frac{1}{2\beta}\int\int d^3x d^3x'\Big[ \frac{\delta (x-x')}{n(x)}-c^{(2)}(|x-x'|;n]\Big]_{n_{eq}(x)}\tilde n(x)\tilde n(x')+\dots,
\end{align}
where $\tilde n(x)\equiv n(x)-n$, and $c^{(2)}(|x-x'|;n]$ is \hl{the} direct correlation function of \hl{a} uniform fluid with density $n=n(x)$.
\hl{The} Fourier transform of the second integral in (\ref{taylor}) gives
\begin{align}
    \frac{1}{2\beta}\int\int d^3x d^3x'\Big[ \frac{\delta (x-x')}{n(x)}-c^{(2)}(|x-x'|;n]\Big]_{n_{eq}(x)}\tilde n(x)\tilde n(x')\\ \nonumber
    =\frac{-1}{2\beta V} \sum_q
    \Big(
    c^{(2)}(q;n]-\frac{1}{n(q)}
    \Big)
    \tilde n(q)\tilde n(-q) \; ,
\end{align}
and comparing with (\ref{e311}) yields
\begin{equation}
    f_0''(n)=\frac{-1}{\beta}(a(n)-\frac{1}{n})\; ,
    \quad \quad
    f_2(n)=\frac{-b(n)}{2\beta},
\end{equation}
where the functions $a(n)$ and $b(n)$ are defined as coefficients of Fourier transform of the Ornstein-Zernike direct correlation function by $c^2(q;n]=a(n(q))+b(n(q))q^2+\dots$.\\
$b(n)$ is evaluated by linear response theory to find that
\begin{equation}\label{lrt}
    f_2(n(x))=\frac{1}{12\beta}\int d^3x' |x-x'|^2 c^{(2)}(|x-x'|;n].
\end{equation}
We can substitute equations (\ref{lrt}) and (\ref{f0}) in (\ref{muin}) and use equilibrium identity $\nabla \mu=0$ to find  the integro-differential equation
\begin{equation}
    \nabla
    \Bigg[
    v(x)+\mu(n(x))-f_2'(n(x))|\nabla n(x)|^2-2f_2(n(x))\nabla^2 n(x)
    \Bigg]_{n(x)=n_{eq}(x)}=0,
\end{equation}
that determines the equilibrium density $\n_{eq}(x)$ in the presence of external potential $v(x)$ given Ornstein-Zernike direct correlation function of uniform fluid $c^{(2)}[n(x),|x-x'|]$.
\section{Conclusion}
We have shown that the variational principle of classical DFT is a special case of applying the method of maximum entropy to construct optimal approximations in terms of those variables that capture the relevant physical information namely, the particle density $n(x)$.
\hl{It is worth emphasizing once again: In this paper we have pursued the purely conceptual goal of finding how DFT fits within the MaxEnt approach to statistical mechanics. The advantage of achieving such an insight is the potential for future applications that lie outside the reach of the current versions of DFT. As an illustration we have discussed the already well-known example of a slowly varying inhomogeneous fluid.}
Future research can be pursued in three different directions: \textbf{i)} To show that the method of maximum entropy can also be used to derive the quantum version of DFT. \textbf{ii)} To approach \hl{the} Dynamic DFT \cite{marconi1}, generalizing the idea to non-equilibrium systems following the theory of maximum caliber \cite{caliber}. \textbf{iii)} To revisit the objective of section \ref{gradient} and construct weighted DFTs \cite{tarazona,rosenfeld} using the method of maximum entropy.







\end{document}